\documentclass[runningheads]{llncs}

\usepackage[T1]{fontenc}
\usepackage[utf8]{inputenc}
\usepackage{graphicx}
\usepackage{amsmath}
\usepackage{amssymb}
\usepackage[numbers]{natbib}
\usepackage{booktabs}
\usepackage{algorithm}
\usepackage{algorithmic}
\usepackage{url}
\usepackage[hidelinks]{hyperref}
\usepackage[small]{caption}
\usepackage{soul}
\usepackage{cleveref}
\usepackage{dsfont}
\usepackage{comment}
\usepackage{tikz}
\usetikzlibrary{positioning,arrows.meta}
\usepackage{todonotes}
\usepackage{thm-restate}

\newcommand{\hs}[1]{\todo[inline,color=orange!30]{Henri: #1}}

\spnewtheorem{observation}{Observation}{\bfseries}{\itshape}

\begin{document}

\title{An Enriched Model of Strategic Voting under Uncertainty}

\titlerunning{An Enriched Model of Strategic Voting}


\author{Henri Surugue\inst{1} \and Sébastien Destercke\inst{2}}
\authorrunning{H. Surugue and S. Destercke}
 \institute{
   LIP6 - Nukkai, Sorbonne University, Paris, France \\ \email{henri.surugue@lip6.fr}
   \and
   CNRS - Heudiasyc, UTC, Compiègne, France \\ \email{sebastien.destercke@cnrs.fr}
 }
 

\maketitle

\begin{abstract}
We present a new strategic voting model where we use uncertainty representation to model preferences. Specifically, we use probability sets as uncertainty representations, together with lower and upper expected utility gains to take strategic decisions. Focusing on belief functions in particular, we demonstrate that this very expressive model includes in one sweep many existing models based on probabilities, sets or incomplete preferences. Additionally, we generalize several well-known convergence results from the literature to this broader representational setting. Furthermore, we illustrate how this model can capture more realistic scenarios for practical applications but also raises theoretical challenges.
\end{abstract}


\section{Introduction}

Strategic voting~\cite{meir2018strategic} arises when voters have an incentive to misreport their true preferences to obtain a more desirable result. The Gibbard-Satterthwaite theorem~\cite{satterthwaite1975strategy} shows that no reasonable voting rule is immune to strategic manipulation. This makes it desirable to formalise the notion of strategic votes and investigate their consequences, first in a single-step, single-voter setting and then in multi-steps, iterative settings. This can be done in various ways: formulating the problem as a strategic game~\cite{meir2010convergence} or as a Bayesian decision problem~\cite{myerson1993theory,hazon2012evaluation}. These approaches raise the question of how voters act strategically as well as the role of polls, and show that understanding how information impacts voters is crucial~\cite{endriss2016strategic,reijngoud2012voter}. 

Several reasons justify the need for an uncertainty-based approach to polls, stemming either from the polling process itself or from the voters themselves. Indeed, building polls, even when done honestly, comes with uncertainty because you cannot ask everyone's preferences (for dishonest polls, see~\cite{mousseau2024do}). Voters' uncertainty can come from unsure preferences or votes, especially if the elections are still far in the future. Whatever the uncertainty sources, they translate into uncertainty about the scores when considering the Plurality rule, as these are the only relevant pieces of information for strategic behavior under this rule. Models of uncertain polls either consider a distribution over possible score vectors \cite{myerson1993theory}, or uncertain neighborhood of scores \cite{meir2017iterative}. Uncertain voters, on their side, are commonly modeled by the means of incomplete preferences \cite{conitzer2011dominating,dey2018complexity,terzopoulou2024iterative}.
Within such models, strategic vote changes have been studied either as a single-step change for a voter~\cite{myerson1993theory,conitzer2011dominating} or as a multi-step iterative procedure~\cite{meir2017iterative,terzopoulou2024iterative,lev2019heuristic}.

Our contributions in this paper are the following:
\begin{itemize}
    \item We consider modeling voting uncertainties by sets of probabilities, a quite expressive model that includes in one sweep set-based and probabilistic approaches and allows voters to better express their uncertainty. Section~\ref{sec:model} introduces the representation, and provides several examples showing how its expressive power can help in poll or voter uncertainty modeling;
    \item In Section~\ref{sec:motivations}, we introduce our decision-theoretic part, providing several decision rules including known ones~\cite{myerson1993theory, conitzer2011dominating}, that we can revisit and interpret within our framework. Such a higher expressiveness could also be useful in descriptive studies, as more voter behaviours can be described by the model;
    \item Finally, our study culminates in Section~\ref{sec:convergence_results} that considers iterative, multi-step frameworks, for which we study convergence results in the style of  Meir~\cite{meir2022strategic}. 
\end{itemize}

\section{Uncertainty model on votes}
\label{sec:model}

We start by introducing the elements of our newly proposed model. As the model itself is one of the main contribution of the paper, we will describe and illustrate it in detail. 

\subsection{Modeling Plurality Elections from the Perspective of Strategic Voting}

Let $N$ be a set of voters where $N=\{1, \ldots, n\}$, and $M$ be a set of candidates where $M=\{1,\ldots,m\}$. 
As we study strategic voting, we assume that $m>2$, as voting rules are susceptible to manipulation only in those cases~\cite{satterthwaite1975strategy}. We assume that each voter $i\in N$ has some true underlying preferences, possibly imperfectly known, over candidates represented by a linear order $\succ_i$ over candidates.
Let $\Pi^m$ be the set of all possible preference orders for $m$ candidates.
In this paper, the winner of an election is determined by the Plurality voting rule where ties are broken lexicographically.
Let $b_i\in M$ denote the ballot of voter $i$ and $b \in M^n$ denote the ballot profile. 
The winner under Plurality of the ballot profile $b$ is $\mathcal{W}(b)\in\arg\max_{x \in M}{s_x(b)} $, where $s_x(b) := |\{i \in N : b_i=x  \}|$ and a lexicographic tie-breaking order, denoted by $\rhd$, is used if necessary. Given a number $n$ of voters, we will denote by $\mathcal{S}=\{s:\sum_{x \in M} s_x(b)=n\}$ the space of possible score vectors. By abuse of notation, we sometimes directly write $\mathcal{W}(s)$ to refer to the winner of a score vector $s$, and more generally use a ballot profile and its corresponding score interchangeably.

While strategic voting is unavoidable~\cite{satterthwaite1975strategy}, a voter can only strategize if she has some information about the election. In plurality, this is relatively straightforward, as knowing the current scores is sufficient to determine whether strategic voting is beneficial. However, the broadcast score is usually subject to uncertainty, as polling institutes can only survey a subset of voters, and voters themselves may be uncertain about their preferences before the election. Two models have met consensus in the literature to handle this problem. 

\subsection{Representing uncertain votes and polls}

\paragraph{Existing representations}
Existing uncertainty representations typically focus on sets and probabilities. The main ones  encountered in strategic voting literature are:
\begin{itemize}
    \item sets $S \subseteq \mathcal{S}$ of possible scores, as for instance in the model of Meir~\cite{meir2017iterative}. This is arguably one of the simplest representation of uncertainty. 
    \item probabilities $p(s)$ of scores over $\mathcal{S}$, which Myerson and Weber~\cite{myerson1993theory} derived from observed votes through a multinomial model. In our case, it suffices to retain that it produces a probability over scores. 
    \item Conitzer et al.~\cite{conitzer2011dominating} builds on the model of Konczak et al.~\cite{konczak2005voting} using the notions of possible and necessary winners, considering that a voter uncertainty is given by a partial order over the $M$ candidates, provided in the form of pairwise comparisons. For plurality rules, this is equivalent to considering the set of maximal elements of the partial order as the potential vote ballot. 
\end{itemize}

While our approach can cope with any subset of $\mathcal{S}$, in further examples, we will often adopt an interval-based notation for sets of scores, where an interval $[\underline{s}_x,\overline{s}^x]$ represents the range of a candidate score. This is a natural assumption that allows the neighborhood to be connected. 
\begin{example}
If $n=5$ and $|M|=3$ with $M=\{a,b,c\}$, the set-valued score $S=([2,3],[1,3],[0,2])$ means that the score of $b$ will be between $1$ (one voter will vote $b$ for sure) and $3$ (up to three voters may vote $b$). It may result from the following 5 imprecise ballots: $\{a\},\ \{a\},\ \{a,b,c\},\ \{b\},\ \{b,c\}$, with the last resulting, for example, from the partial preference $b \succ_i a$. \end{example}
Recall that whenever we use this interval representation, the sum of the scores $s = (s_1,\ldots,s_m)$ picked within $[\underline{s}_i,\overline{s}^i]$ still sums up to $n$, as it is a subset of $\mathcal{S}$.

\paragraph{Probability sets} In our proposal, uncertainty on votes is given by \emph{convex probability set}~$\mathcal{P}$ with probabilities defined over the space~$\mathcal{S}$, either derived from poll results or voters uncertainty. The next example describe a very simple situation that this model can capture but that sets or probabilities alone cannot.

\begin{example}\label{ex:complex_credal}
Assume we have two voters and two candidates\footnote{Other examples will consider more than 2.} $\{a,b\}$. Each voter state that she is more willing to vote for $a$ rather than $b$, while being unsure. This can be translated by $\mathcal{P}_i=\{p_i ~|~ p_i(\{a\}) \geq p_i(\{b\})\}$ for voter $i$. This information cannot be faithfully captured by a set nor by a precise probability.

Assuming voters are independent, these uncertain votes then induce a set of possible probabilities over $\mathcal{S}$, with $P(\{(2,0)\})=p_1(\{a\})p_2(\{a\})$, $P(\{(1,1)\})\!\!=\!\!p_1(\{a\})p_2(\{b\})+p_1(\{b\})p_2(\{a\})$, $P(\{(0,2)\})\!=\!p_1(\{b\})p_2(\{b\})$. Given the constraints on $\mathcal{P}_i$, one can then deduce that $P(\{(2,0)\})\!\geq\! P(\{(0,2)\})$ and $P(\{(1,1)\})\!\geq\! P(\{(0,2)\})$, hence that $a$ will most likely be elected, but not whether $P(\{(2,0)\})\!\geq \!P(\{(1,1)\})$, meaning that the most probable score is not known for sure. 
\end{example}

This set can be interpreted as a collection of plausible probability distributions. 
This very generic model  includes the previously cited one:
\begin{itemize}
    \item a \textbf{set} $S$ of scores, either given~\cite{meir2017iterative} or resulting from incomplete preferences~\cite{conitzer2011dominating} is captured by the set $\mathcal{P}_S=\{P: P(S)=1\},$. In particular, if $S$ is a neighborhood of the true score $s$ with respect to a distance $d$ that we leave unspecified at this stage, then we recover Meir's~\cite{meir2017iterative} representation.
\item a \textbf{probability} $p$ over $\mathcal{S}$ simply amounts to considering $\mathcal{P}=\{p\}$, i.e., the set reduces to a singleton. This means that any framework resulting in such a probability can be easily captured by our setting, as is the case for Myerson~\cite{myerson1993theory} 
\end{itemize}

Manipulating generic sets of probabilities, whether one directly gives a set $\mathcal{P}$ on $\mathcal{S}$ or builds it from individual voter uncertainties $\mathcal{P}_i$, gives us a very expressive framework to describe uncertainties, but is not very computationally friendly nor easy to represent. This is why we detail some specific models of interest~\cite{destercke2014special}, focusing on belief functions and their special cases. 

\begin{remark} Belief functions offer practical advantages, but some natural assessments need the full language of convex sets of probabilities, such as qualitative comparisons~\cite{miranda2015extreme}. Example~\ref{ex:complex_credal} is of this kind. 
\end{remark}

\paragraph{Belief functions}

A belief function over $\mathcal{S}$ consists in defining a positive mass function $\mathcal{M}:\mathcal{S} \to [0,1]$ that sums to one, i.e.,  $\sum_{S \subseteq \mathcal{S}, S \neq \emptyset }\mathcal{M}(S)=1$. From such a mass $\mathcal{M}$, we define two bounds over events $A \subseteq \mathcal{S}$ that are defined as
\begin{gather}
    \underline{P}(A)=\sum_{S \subseteq A} \mathcal{M}(S), \quad \overline{P}(A)=\sum_{S \cap A \neq \emptyset} \mathcal{M}(S)
\end{gather}
and from which one can define a corresponding set of possible probabilities
\begin{equation}
    \mathcal{P}_{\mathcal{M}}=\{P: \forall A, \underline{P}(A) \leq P(A) \leq \overline{P}(A)\}
\end{equation}

Belief functions are still quite expressive, and notably include probabilities and sets as special cases: probabilities correspond to masses bearing only on singletons, while a set $\mathcal{S}$ correspond to the mass $m(\mathcal{S})=1$.

\begin{example}
    Consider a belief function over the set of possible scores $\mathcal{S}$. For example, one can define $\mathcal{M}(\{(1,1,1),(0,2,1)\})=\mathcal{M}(([0,1],[1,2],1))=1.$
    This mass function corresponds to a set and describes two voters that are certain to vote for $b$ and $c$, while one voter is hesitating between $a$ and $b$.
    \end{example}

\begin{remark}\label{rk:independance} We can also represent each voter's ballot using a mass function $\mathcal{M}_i$ for voter $i$ defined over $M$, as we did to recover Conitzer's model~\cite{conitzer2011dominating}. 
From these, one can build a joint mass function over $M^n$ (hence over $\mathcal{S}$) : if $E_i$ are subsets of $M$ having positive masses for voter $i$, the joint mass given over $E_1 \times \ldots \times E_n$ is simply $\mathcal{M}(E_1 \times \ldots \times E_n)=\prod_{i=1}^N \mathcal{M}_i(E_i)$, which corresponds to assuming independence between agents~\cite{smets1992concept}. As each imprecise ballot can be mapped to a set of possible score vectors, we assume in the rest models on the score space, that is, mass functions $\mathcal{M} : 2^\mathcal{S} \to [0,1]$.
\end{remark}

A key concept in belief functions that we will use is the \emph{pignistic probability} $p^*$ of $\mathcal{M}$, which selects a single probability from $\mathcal{P}_{\mathcal{M}}$. It is defined as
$
p^*(s) = \sum_{S \ni s} \frac{\mathcal{M}(S)}{|S|},
$
and corresponds to the Shapley value of $\underline{P}(A)$ viewed as a game, also reflecting a Laplacian assumption within each focal set $S$. It is also a vertex gravity center of $\mathcal{P}_{\mathcal{M}}$, hence a meaningful precise representative of $\mathcal{P}_{\mathcal{M}}$, if one must be selected.
Also, if the mass $\mathcal{M}$ is a probability $p$, then $p^* = p$.  

As mentioned earlier, belief functions are quite useful to model uncertain polls~\cite{kreiss2020undecided}, especially if we allow voters to express their opinion as sets of possible candidates, just as in the work of Conitzer et al.~\cite{conitzer2011dominating} once incomplete preferences are transformed into possible ballots. 

\begin{example}
Consider the case of single-peaked preferences~\cite{black1958theory} where candidates $\{a,b,c,d,e\}$ are ordered with a left-right axis. One could then weaken the result of a poll by considering that each voter could also vote for the nearest candidates on each side (votes for $a$ become $\{a,b\}$, etc.). Imagine that an initial poll of 100 voters gives the initial score $(10,30,10,30,20)$ and 10 voters to $\{a,b\}$, leading to the belief function $\mathcal{M}([0,40],[0,50],[0,70],[0,60],[0,50])=1$. It is clear that the resulting information is quite weak, hence it would make sense to use a less severe weakening process, e.g., assuming that a voter $i$ initially voting for $b$ becomes $\mathcal{M}_i(\{b\})=\alpha$, $\mathcal{M}_i(\{a,b,c\})=1-\alpha$ 
\end{example}

Before stepping to decision, we will now present and illustrate two particular cases of belief functions of interest here, in particular for the convergence results of Section~\ref{sec:convergence_results}. 

\paragraph{Necessity measures}
A necessity measure is a belief function with mass on nested elements, i.e., $m(S_i)\neq 0$ and $m(S_j)\neq 0$ only if $S_i \subseteq S_j$ or $S_j \subseteq S_i$.

\begin{example}\label{ex:hesitating_voter}
Consider a voter $i$ providing the following information about $\{a,b,c\}$: "I will not vote for $c$, and hesitate between $a$ and $b$, with higher chances to vote for $a$". A model for such an opinion is $\mathcal{M}^{i}(\{a\})=0.5$ and $\mathcal{M}^{i}(\{a,b\})=0.5$, meaning that $[\underline{P}(\{a\}),\overline{P}(\{a\})]=[0.5,1]$ and $[\underline{P}(\{b\}),\overline{P}(\{b\})]=[0,0.5]$. 
    \end{example}



    
\begin{example}\label{ex:nested_neigh}
Consider candidates $\{a,b,c\}$ and the initial poll $s^*=(0,2,1)$ as a result on 3 voters. It is natural to consider that $s^*$ is the most plausible result, but that nearby results are plausible as well, albeit less as they are further away from  $s^*$. Building neighborhoods $S^*_r$ by considering $r$ the maximal number of voters changing their votes (so that the sum of votes remains the same), we get $S^*_1=\{s^*\} \cup \{(0,2,1)(1,1,1),(0,1,2),(1,2,0)\}$ and $S^*_2=S^*_1 \cup \{(2,0,1)(1,0,2),(0,0,3),(2,1,0)\}$. One can then consider the mass $\mathcal{M}(\{s^*\})=\alpha$, $\mathcal{M}(S^*_1)=\beta$ and $\mathcal{M}(S^*_2)=1-\alpha-\beta$ with $\alpha \geq \beta \geq (1-\alpha-\beta)$ to denote a decreasing trust of capturing the true results as we get further away from $s^*$. 
\end{example}


These examples illustrate the kind of model we will consider to extend Meir et al.'s results~\cite{meir2017iterative} convergence results in Section~\ref{sec:convergence_results}. 

\paragraph{Inner measures}

An inner measure is a belief function~\cite[Ch. 2.]{denneberg2013non}, where masses are given on a partition of the space, that is $m(S_i)\neq 0$ and $m(S_j)\neq 0$ if $S_i \cap S_j=0$ and $\cup_{S, m(S)>0} S=\mathcal{S}$ (or a subset of interest).


While necessity measures are well-suited to equip set-based approaches with gradual notions of enlargement (due to the nestedness), inner measures are well-suited to extend probabilistic models, as they rely on a partition. They may therefore provide further insights about probabilistic models. This is reflected in Figure~\ref{fig:models}. 

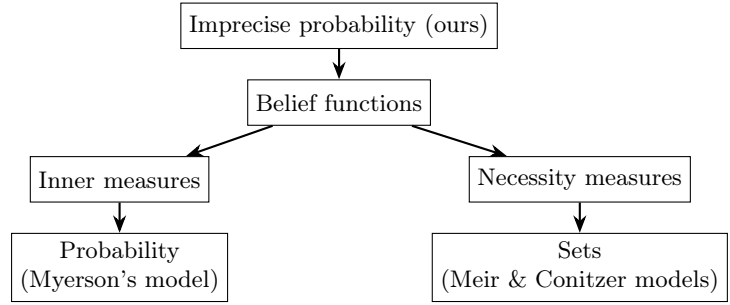
\begin{figure}
\centering
\begin{tikzpicture}[
    node distance=1cm,
    box/.style={rectangle, draw, minimum width=1.8cm, minimum height=0.6cm, align=center, font=\footnotesize},
    arrow/.style={-Stealth, thick}
]
\node[box] (imprecise) at (-0.5,0) {Imprecise probability (ours)};
\node[box, below=0.4cm of imprecise] (belief) {Belief functions};
\node[box, below left=0.4cm and 0.5cm of belief] (inner) {Inner measures};
\node[box, below right=0.4cm and 0.5cm of belief] (possibility) {Necessity  measures};
\node[box, below=0.4cm of possibility] (sets) {Sets\\(Meir \& Conitzer models)};
\node[box, below=0.4cm of inner] (probability) {Probability\\(Myerson's model)};
\draw[arrow] (imprecise) -- (belief);
\draw[arrow] (belief) -- (inner);
\draw[arrow] (inner) -- (probability);
\draw[arrow] (belief) -- (possibility);
\draw[arrow] (possibility) -- (sets);
\end{tikzpicture}
\caption{Overview of Our Uncertainty Model and Its Special Cases}
\label{fig:models}
\end{figure}

\subsection{Relations to existing models}\label{sec:existing_uncertainty}

\paragraph{Meir's model:} \cite{meir2017iterative} defines a neighborhood of scores with respect to a distance $d$ : $S_{r}(s)=\{s'\in \mathcal{S} \; | \; d(s,s') \leqslant r\}$, $d$ can be for instance some $\ell_p$ (\cite{meir2017iterative} takes the $\ell_1$ distance). However, Meir~\cite{meir2017iterative} requires one to fix $r$ and restrict uncertainty to a set. Our models can circumvent this limit by considering multiple $r$ values and by decreasing the plausibility of successive sets, as in Example~\ref{ex:nested_neigh}. 

\begin{example}\label{ex1}
    Consider three candidates $\{a,b,c\}$, $n=30$ the number of voters and the broadcast score being $s=(10,9,11)$.
    We will define beliefs on the neighborhood of scores directly on $\mathcal{S}$, $\mathcal{M}(S_1)=0.5$, $\mathcal{M}(S_2)=0.3$ and $\mathcal{M}(S_3)=0.2$ with $S_1 \subset S_2 \subset S_3$, hence this particular structure is also a necessity measure.
    \end{example}


This belief function, which is a necessity measure, models a certainty decreasing with respect to the distance to the broadcast poll $s$. They are ideal candidate to extend Meir's approach~\cite{meir2017iterative}. 


We may also model the same idea with a more probabilistic flavour. In this case, one can consider the partition generated by the $S_i$ sets, and associate a probability mass to each member of this partition, generating a so-called inner measure~\cite[Ch. 2.]{denneberg2013non} on the initial space $\mathcal{S}$. An interest of this model is that it is a probabilistic one, albeit defined on an algebra whose atoms are non-singleton sets of scores. 

\begin{example}\label{ex2}
    Consider three candidates $\{a,b,c\}$, $n=30$ the number of voters and the broadcast score being $s=(10,9,11)$.
    We will define uncertain information directly on $\mathcal{S}$, $\mathcal{M}(S_1)=0.5$, $\mathcal{M}(S_2\setminus S_1)=0.3$ and $\mathcal{M}(S_3\setminus S_2)=0.2$.
    \end{example}
    
\paragraph{Conitzer's model} Incomplete preferences of voters are naturally embedded in our approach, as they give rise to imprecise ballots, which can be mapped to a belief function. The next example illustrates this. 

\begin{example}\label{ex3}
  Consider three candidates $\{a,b,c\}$, $n=3$ the number of voters and the following preferences given by voters: $b \succ_1 c \succ_1 a \;;\; c \succ_2 a \succ_2 b \;;\; a \succ_3 c$
where voter 3 is sure that she prefers $a$ to $c$, but nothing else. In all completions $a \succ c \succ b$, $b \succ a \succ c$, or $a \succ b \succ c$, $c$ is never ranked first, whereas the top-ranked candidate may be either $a$ or $b$. This is naturally modelled by the belief function over $\mathcal{S}$ defined as $\mathcal{M}(\{(1,1,1),(0,2,1)\})=\mathcal{M}(([0,1],[1,2],1))=1$. 
    \end{example}


\section{Single-step decision model}
\label{sec:motivations}

We now introduce our decision model, illustrate it and link it with the previous single-step manipulation works we mentioned previously.

\subsection{The decision model}
\label{sec:decision_model}
 
From a probability set $\mathcal{P}$ on $\mathcal{S}$ and some utility $u:\mathcal{S} \to \mathbb{R}$, we define the lower expectation
\begin{equation}\label{eq:expectation}\underline{\mathbb{E}}_{\mathcal{P}}(u)=\inf_{P \in \mathcal{P}}\mathbb{E}_P(u)\end{equation}
where $\mathbb{E}_P$ is the standard expectation operator. Upper expectation $\overline{\mathbb{E}}_{\mathcal{P}}(u)$ can be defined likewise, taking a $\sup$ over $\mathcal{P}$. Lower and upper expectations form the basis of decision making with imprecise probabilities, and we will use them to model strategic decisions and manipulations. These lower and upper expectations should be interpreted from the voter's perspective as the worst and best expected scenarios regarding the voting situation, based on their perception of uncertainty and their individual preferences. 

In the case of belief functions, Equation~\eqref{eq:expectation} as well as the upper expectation can be easily solved by adopting the following formula using $\mathcal{M}$, i.e., 
\begin{gather}\label{eq:lower_expe_bel}\underline{\mathbb{E}}_{\mathcal{M}}(u)=\sum_{S \subseteq \mathcal{S}} \mathcal{M}(S) \inf_{x \in S} u(x), \quad \overline{\mathbb{E}}_{\mathcal{M}}(u)=\sum_{S \subseteq \mathcal{S}} \mathcal{M}(S) \sup_{x \in S} u(x). \end{gather}

Let us now encode the strategic behavior of voters under uncertainty. To do so, let us define a function $st(s,a_i,a'_i): \mathcal{S} \times M^2 \to \mathcal{S}$ that sends back an updated score if voter $i$ moves his vote from $a_i$ to $a'_i$. To simplify reading, we will use the shorthand $s_{a_i \to a'_i}:=st(s,a_i,a'_i)$ for the update of score $s$ after a strategic move.
Then, let $u_i: M \to \mathbb{R}$ be for each voter $i$ the associated utility that reflects how satisfied the voter is with a winner in $M$. We will stick to this simple definition depending only on $M$, as it is sufficient for our purpose. Also note that by Debreu et al.~\cite{debreu1954representation}, we can always build a utility function from the strict linear orders $\succ_i$. Therefore, to understand how voters may change their vote, we need to evaluate the benefit of moving from $a_i$ to $a'_i$ that we denote
\begin{equation}\label{eq:util_mov} u_i(a_i’|a_i,s),\end{equation}
that is the utility of moving to $a_i'$, given the initial state of affairs $s$ and the current vote $a_i$. In particular, if we are certain about our current score vector $s$, it suffices for~\eqref{eq:util_mov} to be positive for the voter to make a strategic move, that is 
$$a_i' \succeq a_i \text{ iff } u_i(a_i’|a_i,s) \geq 0$$
where $\succeq$ denotes here that action of voting for $a'_i$ is preferable to voting for $a_i$. 

However, in our model our agents receive a probability set $\mathcal{P}$ over $\mathcal{S}$, and we will define four decision criteria (denoted $DC$) under $\mathcal{P}$.

\begin{itemize}
    \item The first criterion that we will consider is a pessimistic criterion to decide whether an action is better than another as follows:
\[a_i' \succ_{pess} a_i \text{ iff } \underline{\mathbb{E}}_{\mathcal{P}}(u_i(a_i’|a_i,\cdot)) \geq 0 \text{ and } \overline{\mathbb{E}}_{\mathcal{P}}(u_i(a_i’|a_i,\cdot)) > 0\]
under uncertainty $\mathcal{P}$. This corresponds to the maximin of Gilboa et al.~\cite{gilboa1989maxmin}.
    \item A pignistic criterion to decide whether an action is better than another:
    $$a_i' \succeq_{pig} a_i \text{ iff } \mathbb{E}_{p^*}(u_i(a_i’|a_i,\cdot)) \geq 0,$$
    where $p^*$ is the pignistic probability. It includes probabilistic decision as a special case.
   \item A mixture criteria to decide whether an action is better than another:
   \[a_i' \succeq_{\alpha,p^*} a_i \text{ iff }  \alpha \cdot \underline{\mathbb{E}}_{\mathcal{M}}(u_i(a_i’|a_i,\cdot))+ (1-\alpha) \cdot \mathbb{E}_{p^*}(u_i(a_i’|a_i,\cdot)) \geq 0\] 
   with $\alpha \in [0,1]$. $\alpha$ can be seen as a level of completeness of the obtained preferences between possible moves, as it allows one to go from a rather partial order ($\alpha=0$) to a complete one ($\alpha=1$).
   \item Another mixture of criteria:
      \[a_i' \succeq_{H(\alpha)} a_i \text{ iff } \alpha \cdot \underline{\mathbb{E}}_{\mathcal{M}}(u_i(a_i’|a_i,\cdot))
   + (1-\alpha) \cdot \overline{\mathbb{E}}_{\mathcal{M}}(u_i(a_i’|a_i,\cdot)) \geq 0\] 
   with $\alpha \in [0,1]$. This corresponds to the well-known Hurwicz criterion~\cite{hurwicz1951optimality} that intends to balance between optimism ($\overline{\mathbb{E}}$) and pessimism ($\underline{\mathbb{E}}$), and that has been justified within a belief function setting by Denoeux et al~\cite{denoeux2020interval}.
        \end{itemize}

Note that for the three last rules, the strict relation $\succ$ corresponds to strict positive values on the right side. Those decision rules are akin to commonly used decision rules within the imprecise probabilistic setting~\cite{troffaes2007decision}. 

        

To sum up, we describe an election under plurality with a lexicographic tie breaking and strategic voting under uncertainties with the tuple $(N,M,\succ,\rhd,\mathcal{P}, (u_i)_{i\in N}, DC)$. Let us  give an example putting our framework at work, where the considered uncertainty is neither a set nor a probability. 

\begin{example}\label{ex4}
    Consider candidates $\{a,b,c\}$ and $n = 3$, with voter 1 as in Example~\ref{ex:hesitating_voter}, while the second and the third voters have certain ballots $b$ and $c$. In this case, the belief function over $\mathcal{S}$ is defined as $\mathcal{M}(\{(1,1,1)\})=0.5,\mathcal{M}(([0,1],[1,2],1))=0.5$. Let us now illustrate one of our decision rules. Consider for instance that the preference of the second voter is given by $b \succ_2 c \succ_2 a$, and a Hurwicz criterion with $\alpha=\frac{1}{3}$, meaning that the weight of optimism is larger. We will consider the following utilities
    \begin{equation}\label{eq:utility_Meir}\forall i \in N, u_i(a_i’|a_i,s) =
    \begin{cases} 
        1 & \text{if}~ \mathcal{W}(s_{a_i \to a'_i}) \succ_i \mathcal{W}(s), \\
        0  & \text{if}~ \mathcal{W}(s_{a_i \to a'_i}) \sim_i \mathcal{W}(s), \\
        -1 & \text{otherwise}. \\
    \end{cases}\end{equation}
    We want to evaluate the deviation from $b$ to $c$ for voter 2, where a deviation simply means a change of ballot. Then, using Eq.~\eqref{eq:lower_expe_bel}, we compute \[\underline{\mathbb{E}}_{\mathcal{M}}(u_2(a_2’|a_2,\cdot)) = \frac{1}{2} \cdot 1 + \frac{1}{2} \cdot (-1)= 0 \;; \;\overline{\mathbb{E}}_{\mathcal{M}}(u_2(a_2’|a_2,\cdot))=\frac{1}{2} \cdot 1 + \frac{1}{2} \cdot 1 = 1.\]
For the precise set $\{(1,1,1)\}$ with mass $0.5$, the move makes $c$ elected instead of $a$, and the utility is a constant $1$ (there is no uncertainty on set $\{(1,1,1)\}$). For the set $([0,1],[1,2],1)$ that also has mass $0.5$, the move is still beneficial for $(1,1,1)$, but would bring negative utility for $(0,2,1)$, as $c$ would be elected instead of $b$, a clear given $b \succ_2 c \succ_2 a$. Thus, we get: \[ \frac{1}{3} \cdot \underline{\mathbb{E}}_{\mathcal{M}}(u_2(a_2’|a_2,\cdot))+ (1-\frac{1}{3}) \cdot \overline{\mathbb{E}}_{\mathcal{M}}(u_2(a_2’|a_2,\cdot))=\frac{1}{3} \cdot 0 + \frac{2}{3} \cdot 1= \frac{2}{3} \geq 0,\]
 concluding that voter 2 will do the deviation from $b$ to $c$ with such a criterion.
    \end{example}

This results in a rich uncertainty modeling and decision framework, in which one can easily separate what comes from the uncertainty model, and what comes from the decision rule. Let us now see how our approaches connect with the existing literature. 



\subsection{Relations with Existing Models}

We show that the models from Meir et al.~\cite{meir2010convergence,meir2017iterative}, Myerson et al.~\cite{myerson1993theory} and Conitzer et al.~\cite{conitzer2011dominating} are specific case of our approach, demonstrating its expressiveness. 


For the first, we consider the model $S_r(s)$ mentioned in Section~\ref{sec:existing_uncertainty}. By abuse of notation, we will use $S_{r}$ when it is clear from the context. Meir~\cite{meir2010convergence} first worked on the special case $r=0$, meaning the poll information is complete for plurality and certain. We retrieve Meir's model ~\cite{meir2010convergence} by considering voters $i \in N, \mathcal{M}_i(S_0)=1$ and the particular utility functions:
\[\forall i \in N, u_i(a_i’|a_i,s)=
\begin{cases} 
  1 & \text{if}~ a'_i=\mathcal{W}(s_{a_i \to a'_i}) \succ_i \mathcal{W}(s),  \\
  0  & \text{if}~ a'_i\neq\mathcal{W}(s_{a_i \to a'_i}) \succ_i \mathcal{W}(s) \text{ or } \mathcal{W}(s_{a_i \to a'_i}) \sim_i \mathcal{W}(s), \\
  -1 & \text{otherwise}, \\
\end{cases}
\]
and the pessimistic decision rule $\succeq_{pess}$. Recall that in~\cite{meir2017iterative}, convergence is  guaranteed and we call these strategic moves ``direct best response" since all strategic moves where a voter deviates are to the new winner.

Meir~\cite{meir2014local} then extended this framework to add uncertainty by considering strictly positive value $r_i>0$ for voter $i$. This comes down in our case to take, $\forall i \in N, \mathcal{M}_i(S_{r_i})=1$, the utility functions given by Equation~\eqref{eq:utility_Meir} and again the pessimistic decision rule $\succeq_{pess}$ to recover Meir~\cite{meir2014local} proposal.

For the probabilistic model of Myerson et al.~\cite{myerson1993theory}, they derive from observed votes a weight vector $q=(q_1,\ldots,q_m)$, and use a multinomial distribution to obtain the probability over score vectors $\mathds{P}(s)$.
The utilities can be chosen arbitrarily provided they satisfy the voter’s preferences $u_i(a_i’|a_i,s) \geq 0  ~\text{iff}~\mathcal{W}(s_{a_i \to a'_i}) \succ_i \mathcal{W}(s).$
Since in the probabilistic case all decision rules we have described in Section~\ref{sec:decision_model} reduce to a classical expectation operator, there is no need to specify a specific behaviour. 

The last model we want to encompass is the one from \cite{conitzer2011dominating} with incomplete preferences. As we already showed in Section~\ref{sec:model}, the incompleteness on each voter's preference can easily be summarized as an uncertainty set on scores that we call $S$ and where $\mathcal{M}(S)=1$. It remains to choose the decision rule with respect to this set. Conitzer et al.\cite{conitzer2011dominating} adopt a pessimistic decision criterion, which amounts to use the utilities of~\eqref{eq:utility_Meir}
and the pessimistic decision rule $\succeq_{pess}$.

Now that we have a fully fledged decision framework allowing us to model whether or not a voter will proceed to a strategic move or manipulation, we study the multi-step case, in particular showing that classical convergence results can be reinterpreted and extended to specific instantiations of our framework.

\section{Multi-step Manipulations}
\label{sec:convergence_results}

In this section we generalize existing results on convergence from Meir~\cite{meir2017iterative} by adding quantitative uncertainty about the score as in Example~\ref{ex1} and Example~\ref{ex2}. Let us recall that a neighborhood of scores $S_r$ with respect to a distance $d$ is defined as in Meir's~\cite{meir2017iterative}: $S_{r}(s)=\{s'\in \mathcal{S} \; | \; d(s,s') \leqslant r\}$, where $d$ can be $\ell_1, \ell_{\infty}$, \ldots as in Meir's work~\cite{meir2017iterative}. Let us consider the $\ell_1$ distance, which is easy to interpret in terms of adding or removing votes, without necessarily preserving the total number of votes, which may vary by one. This behavior can be easily explained by considering a voter who chooses to abstain. For each voter $i$, we will let $r_i$ be the support of their uncertainty, so that $S_{r_i}$ is the biggest set of voter $i$'s uncertainty. We now formalize the notion of equilibrium in this model through the following definition.

\begin{definition}
    An equilibrium is a situation where no voter has an incentive to deviate from its ballot, i.e. $\forall i \in N, \nexists~ a_i' ~\text{such that}~ a_i' \succ_{DC} a_i$. We assume voters will only make a strategic move if their preferences are strict.
\end{definition}

\subsection{Extending Meir Framework}
Let us consider a belief function for voter $i$, $\mathcal{M}^i:\mathcal{S} \to [0,1]$ that sums to one, i.e.,  $\sum_{S \subseteq S_{r_i} , S \neq \emptyset }m^i(S)=1$. This notation for belief functions indexed by voters should not be confused with that of Remark~\ref{rk:independance}, where the question was to model each ballot independently, while here we consider that the scores come from a unique poll, but that each voter may be more or less skeptical about its result.
We will consider two particular cases of interest. The first is the case of nested sets, i.e. $\forall i \in N, \forall k \geqslant 1, m^i(S_k)=\beta_k^i$, with $\sum_{k=1}^{r_i} \beta_k^i=1$. Indeed, we might want to describe the fact that the belief mass depends on the distance to the true score. Second, we look at the case of \emph{partitioned} belief function, i.e. $\forall i \in N, \forall k \geqslant 2, m^i(S_k \setminus S_{k-1} )=\beta_k^i$, with $\sum_{k=1}^{r_i} \beta_k^i=1$ and $\forall i \in N, m^i(S_1)=\beta_1^i$.
In both cases, we will assume decreasing $(\beta_k^i)_{1\leqslant k \leqslant r_i}$, meaning that our evidence decreases as we get further from the observed $s$. For the rest of this section, we consider utilities given by Equation~\eqref{eq:utility_Meir}.


\begin{restatable}{theorem}{cvgone}\label{cvgone}
    Voters considering uncertainty given by a nested or a partitioned decreasing belief function, and making strategic decisions according to either pessimistic ($\succeq_{pess}$) or mixed ($\succeq°_{\alpha,p^*}, \succeq_{H(\alpha)}$) decision rules with $\alpha$ large enough, will converge to an equilibrium.
\end{restatable}

\begin{proof}
    At first, we will take $\forall i \in N, \forall k \geqslant 1, \mathcal{M}(S_k)=\beta_k^i$, with $\sum_{k=1}^{r_i} \beta_k^i=1$. We consider a feasible move $a_i \to a_{i'}$.
    We remark that the hypothesis that $\underline{\mathbb{E}}_{\mathcal{M}}(u_i(a_i’|a_i,s))\geqslant 0$ implies that either:
    \begin{itemize}
        \item $\forall s \in S_{r_i}, u_i(a_i’|a_i,s) \neq -1,$
        \item $\exists \tilde{r} \in [0,r_i], \forall s \in S_{\tilde{r}}, u_i(a_i’|a_i,s) = 1 $
    \end{itemize}
    since, if some neighborhoods have a negative utility, this must be compensated by a positive contribution, which happens only if all moves within a neighborhood are positive. However, the second case is impossible, for the reason that if we allow a voter to transfer its vote to another candidate ($r=1$), then there is a situation for which the voter is not pivotal, and therefore $\exists s \in S_{1}, u_i(a_i’|a_i,s) = 0$, showing that the second case never happens. Therefore, if $\underline{\mathbb{E}}_{\mathcal{M}}(u_i(a_i’|a_i,s))\geqslant 0$, this means that $\forall s \in S_{r_i}, u_i(a_i’|a_i,s) \in \{ 0,1\}$, with at least one $s$ giving the null value.
    
    
    Second, note that we cannot have $\overline{\mathbb{E}}_{\mathcal{M}}(u_i(a_i’|a_i,s))>0$ if $\forall s \in S_{r_i}, u_i(a_i’|a_i,s)= 0$, and there must be a situation for which making this strategic move is a local dominance move, meaning it verifies Theorem 4 from Meir~\cite{meir2017iterative}.
    
    We do exactly the same reasoning for the second type of belief functions, i.e., $\forall i \in N, \forall k \geqslant 2, \mathcal{M}_i(S_k \setminus S_{k-1} )=\beta_k^i$, with $\sum_{k=1}^{r_i} \beta_k^i=1$ and $\forall i \in N, \mathcal{M}_i(S_1)=\beta_1^i$. For the mixed decision, one can prove that the decision behaves as the pessimistic one if $\alpha$ is large enough.
    \end{proof}

We now want to go a step further by showing the convergence can hold even with a less, but still pessimistic behavior, namely the Hurwicz criterion with $\alpha>\frac{1}{2}$. This allows us to consider negative outcomes in the uncertainty neighborhood of scores. Let us consider belief functions as follows: $\forall i \in N, \forall k \geqslant 1, m^i(S_k)=\beta_k^i$, with $\sum_{k=1}^{r_i} \beta_k^i=1$ 

\begin{restatable}{theorem}{cvgtwo}\label{theo:cvgtwo}
    Voters considering a decreasing belief function 
    around the true score and making strategic votes according to a sufficiently pessimistic Hurwicz criterion  (i.e. $\alpha>\frac{1}{2}$) will converge to an equilibrium.
\end{restatable}

\begin{proof}[sketch]
   Hurwicz can be rewritten as  
   $\sum_{i=1}^{r_i}\mathcal{M}(S_i) [\alpha \inf_{s\in S} (u_i(a_i’|a_i,s)) + (1-\alpha)\cdot \sup_{s\in S} (u_i(a_i’|a_i,s)) ]$
   and we show that if $\alpha >\frac{1}{2}$, such an equation is positive only if we are in averaging over situations that are similar from section VIII in Meir~\cite{meir2015plurality} which tells us that the convergence holds for any level of uncertainty $r_i$ and any starting point (even non-truthful states) for local dominance strategic behaviors.
\end{proof}

\subsection{The Case of Pignistic Probability}    

Let us now consider the case of a pignistic probability on a neighborhood of scores $S_{r_i}$, which is equivalent to having a uniform distribution $\mathcal{U}$ over all single scores within $S_{r_i}$ for any voter $i$. If $S_{r_i}$ is our single set with positive mass, the corresponding strategic behaviour is equivalent having the following criterion in the strict case: 
\[a_i' \succ_{pig} a_i \text{ iff }  \operatorname{card}(s\in S_{r_i} ~|~ u_i(a_i’|a_i,s)=1)- \operatorname{card}(s\in S_{r_i} ~|~ u_i(a_i’|a_i,s)=-1) > 0\]
where $\operatorname{card}$ is the cardinality of the set.


\begin{restatable}{proposition}{counterexemple}\label{prop:counterexemple}
    With a uniform pignistic criterion on a neighborhood of scores of size 1 (with respect to the $\ell_1$ distance), convergence is not guaranteed.
\end{restatable}

\begin{proof}[sketch]
If we consider the preferences
\begin{gather*}
a \succ_1 b \succ_1 c \succ_1 d \; ; \; c \succ_2 a \succ_2 b \succ_2 d \; ; \;  c \succ_3 a \succ_3 b \succ_3 d  \\
c \succ_4 d \succ_4 a \succ_4 b \; ; \; a \succ_5 c \succ_5 d \succ_5 b \; ; \;  d \succ_6 b \succ_6 a \succ_6 c  \\
c \succ_7 b \succ_7 d \succ_7 a \; ; \; d \succ_8 b \succ_8 a \succ_8 c \; ; \;  b \succ_9 d \succ_9 c \succ_9 a  \\
b \succ_{10} d \succ_{10} c \succ_{10} a 
\end{gather*}
of those ten voters over four candidates, we can show that Voter $8$ will move from $d$ to $a$ with a cardinal difference of 1, voter $2$ will move from $c$ to $a$ with a cardinal difference of 2, voter $8$ will move from $a$ to $d$ with a cardinal difference of 1 and finally voter $2$ will move from $a$ to $c$ with a cardinal difference of 3. This creates a cycle, which prevents convergence.
\end{proof}

This example is to be put in perspective with our result of Theorem~\ref{cvgone} about partitioned belief functions, which indicates that with the right (coarse) algebra, it is possible to consider a probability measure and a corresponding decision rule such that convergence holds, meaning that the counter-example is not so much about having a probabilistic model itself than about the voters being too optimistic in their movement, hinting also at the fact that "optimistic" decision rules 
are unlikely to lead to convergence of voting behaviors. Note that the distance here is the $\ell_1$ distance, and that the number of voters does not sum to $n$. However, this can easily be interpreted as meaning that one voter does not participate, and is therefore still very reasonable. Moreover, similar examples can be found for other distances, including distances for which the number of voters sums to $n$. 

\section{Conclusion and Future Works}
\label{sec:discussion}
In this paper, we provide new tools to model uncertainty in strategic voting. We think that our work provides substantially new views on preference modeling in voting theory, and in particular in strategic voting models. In particular, we have shown that they capture standard settings in a unifying framework, and allow one to represent new uncertainty scenarios that were not captured by previous models. We have illustrated, through many examples, that this model can account for uncertainty arising both from polls and from the voters themselves. Moreover, we establish new convergence results in our framework that generalize existing ones, and we provide a key counterexample that points out the limits of these convergence results.

This work opens up several promising research directions, both in strategic voting models and in preference modeling in voting theory more broadly. A first important question is how incorporating a more realistic model of uncertainty may influence election outcomes compared to classical iterative voting models~\cite{meir2022strategic}. Another avenue for future research is to design empirical studies aimed at identifying the parameters of our model, in order to represent the uncertainty present in both polls and voters’ preferences more accurately in real-world elections. We also believe that the descriptive and expressive power of the framework is an advantage: it can formally capture more behaviors than other frameworks, and can ensure smooth transitions between set and probabilistic models. A challenge will however be to propose parametric models where the number of parameters does not explode. Suppose we are able to do so, it would allow us to reframe classical questions in iterative voting, such as whether equilibria can be computed efficiently~\cite{rabinovich2015analysis}, how an external agent might manipulate the outcome~\cite{baumeister2020manipulation}, or whether it increases social welfare~\cite{kavner2021strategic}. One might also ask how such a model could be adapted to other voting rules. Finally, we believe this modeling can offer a new perspective on other voting problems, such as possible and necessary winners or elicitation of voters’ preferences, by capturing uncertainty in a more quantitative flavor.

\section*{Acknowledgement}
This work is supported under the France 2030 program by the ANR-23-IACL-0007 grant (AI Cluster PostGenIA).

\bibliographystyle{splncs04}
\bibliography{ADT}

@incollection{grabisch2016decision,
  title={Decision Under Risk and Uncertainty},
  author={Grabisch, Michel},
  booktitle={Set Functions, Games and Capacities in Decision Making},
  pages={281--323},
  year={2016},
  publisher={Springer}
}

@article{troffaes2007decision,
  title={Decision making under uncertainty using imprecise probabilities},
  author={Troffaes, Matthias CM},
  journal={International journal of approximate reasoning},
  volume={45},
  number={1},
  pages={17--29},
  year={2007},
  publisher={Elsevier}
}

@article{miranda2015extreme,
  title={Extreme points of the credal sets generated by comparative probabilities},
  author={Miranda, Enrique and Destercke, S{\'e}bastien},
  journal={Journal of Mathematical Psychology},
  volume={64},
  pages={44--57},
  year={2015},
  publisher={Elsevier}
}

@article{denoeux2020interval,
  title={An interval-valued utility theory for decision making with Dempster-Shafer belief functions},
  author={Denoeux, Thierry and Shenoy, Prakash P},
  journal={International Journal of Approximate Reasoning},
  volume={124},
  pages={194--216},
  year={2020},
  publisher={Elsevier}
}

@techreport{hurwicz1951optimality,
  title={Optimality criteria for decision making under ignorance},
  author={Hurwicz, Leonid},
  year={1951},
  institution={Cowles Commission discussion paper, statistics}
}

@book{denneberg2013non,
  title={Non-additive measure and integral},
  author={Denneberg, Dieter},
  volume={27},
  year={2013},
  publisher={Springer Science \& Business Media}
}

@article{gilboa1989maxmin,
  title={Maxmin expected utility with non-unique prior},
  author={Gilboa, Itzhak and Schmeidler, David},
  journal={Journal of mathematical economics},
  volume={18},
  number={2},
  pages={141--153},
  year={1989},
  publisher={Elsevier}
}

@article{terzopoulou2024iterative,
  title={Iterative voting with partial preferences},
  author={Terzopoulou, Zoi and Terzopoulos, Panagiotis and Endriss, Ulle},
  journal={Artificial Intelligence},
  volume={332},
  pages={104133},
  year={2024},
  publisher={Elsevier}
}

@book{black1958theory,
  title={The theory of committees and elections},
  author={Black, Duncan},
  year={1958},
  publisher={Cambridge University Press}
}

@inproceedings{lev2019heuristic,
  title={Heuristic voting as ordinal dominance strategies},
  author={Lev, Omer and Meir, Reshef and Obraztsova, Svetlana and Polukarov, Maria},
  booktitle={Proceedings of the 33rd AAAI Conference on Artificial Intelligence (AAAI 2019)},
  volume={33},
  number={01},
  pages={2077--2084},
  year={2019}
}

@inproceedings{konczak2005voting,
  title={Voting procedures with incomplete preferences},
  author={Konczak, Kathrin and Lang, J{\'e}r{\^o}me},
  booktitle={Proceedings of the Multidisciplinary Workshop on Advances in Preference Handling (IJCAI 2005)},
  volume={20},
  pages={12},
  year={2005}
}

@inproceedings{smets1992concept,
  title={The concept of distinct evidence},
  author={Smets, Philippe and Kennes, R},
  booktitle={Proceedings of the Information Processing and Management of Uncertainty (IPMU 1992)},
  pages={789--794},
  year={1992}
}

@article{destercke2014special,
  title={Special Cases},
  author={Destercke, S{\'e}bastien and Dubois, Didier},
  journal={Introduction to Imprecise Probabilities},
  number={chapter 4},
  pages={79--91},
  year={2014},
  publisher={Wiley}
}

@article{debreu1954representation,
  title={Representation of a preference ordering by a numerical function},
  author={Debreu, Gerard and others},
  journal={Decision processes},
  volume={3},
  pages={159--165},
  year={1954},
  publisher={Wiley New York}
}

@book{meir2022strategic,
  title={Strategic voting},
  author={Meir, Reshef},
  year={2022},
  publisher={Springer Nature}
}

@inproceedings{baumeister2020manipulation,
  title={Manipulation of opinion polls to influence iterative elections},
  author={Baumeister, Dorothea and Selker, Ann-Kathrin and Wilczynski, Ana{\"e}lle},
  booktitle={Proceedings of the 19th International Conference on Autonomous Agents and Multiagent Systems (AAMAS 2020)},
  pages={132--140},
  year={2020}
}

@article{kavner2021strategic,
  title={Strategic behavior is bliss: iterative voting improves social welfare},
  author={Kavner, Joshua and Xia, Lirong},
  journal={Proceedings of the 35th Conference on Neural Information Processing Systems (NeurIPS 2021)},
  volume={34},
  pages={19021--19032},
  year={2021}
}

@article{meir2018strategic,
  title={Strategic voting},
  author={Meir, Reshef},
  journal={Synthesis lectures on artificial intelligence and machine learning},
  volume={13},
  number={1},
  pages={1--167},
  year={2018},
  publisher={Morgan \& Claypool Publishers}
}

@incollection{meir2017iterative,
  title={Iterative voting},
  author={Meir, Reshef},
  booktitle={Trends in Computational Social Choice},
  editor={U. Endriss},
  chapter={4},
  pages={69--86},
  year={2017},
  publisher={AI Access}
}

@article{satterthwaite1975strategy,
  title={Strategy-proofness and Arrow's conditions: Existence and correspondence theorems for voting procedures and social welfare functions},
  author={Satterthwaite, Mark Allen},
  journal={Journal of economic theory},
  volume={10},
  number={2},
  pages={187--217},
  year={1975},
  publisher={Elsevier}
}

@inproceedings{meir2014local,
  title={A local-dominance theory of voting equilibria},
  author={Meir, Reshef and Lev, Omer and Rosenschein, Jeffrey S},
  booktitle={Proceedings of the fifteenth ACM conference on Economics and computation (ACM 2014)},
  pages={313--330},
  year={2014}
}

@inproceedings{rabinovich2015analysis,
  title={Analysis of equilibria in iterative voting schemes},
  author={Rabinovich, Zinovi and Obraztsova, Svetlana and Lev, Omer and Markakis, Evangelos and Rosenschein, Jeffrey},
  booktitle={Proceedings of the 29th AAAI conference on artificial intelligence (AAAI 2015)},
  volume={29},
  number={1},
  year={2015}
}

@inproceedings{meir2015plurality,
  title={Plurality voting under uncertainty},
  author={Meir, Reshef},
  booktitle={Proceedings of the 29th AAAI Conference on Artificial Intelligence (AAAI 2015)},
  volume={29},
  number={1},
  year={2015}
}

@inproceedings{mousseau2024do,
  title={Do we Care about Poll Manipulation in Political Elections?},
  author={Mousseau, Vincent and Surugue, Henri and Wilczynski, Anaëlle},
  booktitle={Proceedings of the 27th European Conference on Artificial Intelligence (ECAI 2024)},
  pages={},
  year={2024}
}

@inproceedings{meir2010convergence,
  title={Convergence to equilibria in plurality voting},
  author={Meir, Reshef and Polukarov, Maria and Rosenschein, Jeffrey and Jennings, Nicholas},
  booktitle={Proceedings of the 24th AAAI conference on artificial intelligence (AAAI 2010)},
  volume={24},
  pages={823--828},
  year={2010}
}

@inproceedings{endriss2016strategic,
  title={Strategic voting with incomplete information},
  author={Endriss, Ulle and Obraztsova, Svetlana and Polukarov, Maria and Rosenschein, Jeffrey S},
  year={2016},
  organization={AAAI Press/International Joint Conferences on Artificial Intelligence}
}

@inproceedings{conitzer2011dominating,
  title={Dominating manipulations in voting with partial information},
  author={Conitzer, Vincent and Walsh, Toby and Xia, Lirong},
  booktitle={Proceedings of the 25th AAAI conference on artificial intelligence (AAAI 2011)},
  volume={25},
  number={1},
  pages={638--643},
  year={2011}
}

@inproceedings{reijngoud2012voter,
  title={Voter Response to Iterated Poll Information},
  author={Reijngoud, A and Endriss, U},
  booktitle={Proceeding of the 4th International Conference on Autonomous Agents and Multiagent Systems (AAMAS 2012)},
  volume={4},
  pages={8},
  year={2012}
}

@article{dey2018complexity,
  title={Complexity of manipulation with partial information in voting},
  author={Dey, Palash and Misra, Neeldhara and Narahari, Yadati},
  journal={Theoretical Computer Science},
  volume={726},
  pages={78--99},
  year={2018},
  publisher={Elsevier}
}

@article{myerson1993theory,
  title={A theory of voting equilibria},
  author={Myerson, Roger B and Weber, Robert J},
  journal={American Political science review},
  volume={87},
  number={1},
  pages={102--114},
  year={1993},
  publisher={Cambridge University Press}
}

@article{hazon2012evaluation,
  title={On the evaluation of election outcomes under uncertainty},
  author={Hazon, Noam and Aumann, Yonatan and Kraus, Sarit and Wooldridge, Michael},
  journal={Artificial Intelligence},
  volume={189},
  pages={1--18},
  year={2012},
  publisher={Elsevier}
}

@inproceedings{kreiss2020undecided,
  title={Undecided voters as set-valued information--towards forecasts under epistemic imprecision},
  author={Kreiss, Dominik and Augustin, Thomas},
  booktitle={Proceedings of the 14th International Conference of Scalable Uncertainty Management (SUM 2020)},
  pages={242--250},
  year={2020},
  organization={Springer}
}


\clearpage
\appendix
\begin{center}
    \LARGE \textbf{Technical Appendix} 
\end{center}

\section{Complete proofs}
\counterexemple*

\begin{proof}
    Here is the counter example with a neighborhood of scores of size 1.
    Let us take the following profile:
\begin{gather*}
a \succ_1 b \succ_1 c \succ_1 d \; ; \; c \succ_2 a \succ_2 b \succ_2 d \; ; \;  c \succ_3 a \succ_3 b \succ_3 d  \\
c \succ_4 d \succ_4 a \succ_4 b \; ; \; a \succ_5 c \succ_5 d \succ_5 b \; ; \;  d \succ_6 b \succ_6 a \succ_6 c  \\
c \succ_7 b \succ_7 d \succ_7 a \; ; \; d \succ_8 b \succ_8 a \succ_8 c \; ; \;  b \succ_9 d \succ_9 c \succ_9 a  \\
b \succ_{10} d \succ_{10} c \succ_{10} a 
\end{gather*}

Voter $8$ will move from $d$ to $a$ with a cardinal difference of 1.
For clarity, let us detail the computation of this first move:
the original score is $s=(2, 2, 3, 3)$, so \[S_1=\{
(2, 2, 3, 3),(1, 2, 3, 3),(2, 1, 3, 3),(2, 2, 2, 3),(2, 2, 3, 2), \] \[ (3, 2, 3, 3),(2, 3, 3, 3),(2, 2, 4, 3),(2, 2, 3, 4)\}\] that corresponds to add/remove one vote. 
When moving from $d$ to $a$, Voter 8 is improved in the first state $(2,2,3,3)$ as it becomes $(3,2,3,2)$, and $a$ is elected instead of $c$, which is better from voter 8's perspective. Voter 8 is also improved for states $(2,1,3,3)$, $(2,2,3,2)$ and $(2,3,3,3)$, 
and deteriorated in states $(2,2,2,3)$, $(2,3,3,3)$ and $(2,2,3,4)$.
Other states are not impacted by this move.
Then voter $2$ will move from $c$ to $a$ with a cardinal difference of 2, voter $8$ will move from $a$ to $d$ with a cardinal difference of 1 and finally voter $2$ will move from $a$ to $c$ with a cardinal difference of 3. This creates a cycle, which prevents convergence. 
    \end{proof}

\cvgtwo*
    
\begin{proof}
    We know that the weights $\beta_k^i$ of our belief functions defined as $\forall i \in N, \forall k \geqslant 1, \mathcal{M}_i(S_k)=\beta_k^i$ are decreasing.
    
    We consider a feasible move $a_i \to a_{i'}$.
    Using the fact that the lower (an upper) expectation is positive homogeneous, i.e., $\alpha \underline{\mathbb{E}}(f)= \underline{\mathbb{E}}(\alpha f)$, we get  
    \begin{align}
    &\alpha \cdot \underline{\mathbb{E}}_{\mathcal{M}}(u_i(a_i’|a_i,s))  + (1-\alpha) \cdot  \overline{\mathbb{E}}_{\mathcal{M}}((u_i(a_i’|a_i,s)) \nonumber \\
    &=\sum_{i=1}^{r_i}\mathcal{M}(S_i) [\alpha \inf_{s\in S} (u_i(a_i’|a_i,s)) + (1-\alpha)\cdot \sup_{s\in S} (u_i(a_i’|a_i,s)) ] \label{eq:Hurw}
    \end{align}
    
    Let us denote by \[\Tilde{u}_i(S,a_i,a'_i)=\alpha \inf_{s \in S} (u_i(a_i’|a_i,s))+ (1-\alpha)\cdot \sup_{s\in S} (u_i(a_i’|a_i,s)) \]
    the term associated to subset $S$, meaning that Equation~\eqref{eq:Hurw} can be rewritten $\sum_{i=1}^{r_i}\mathcal{M}(S_i) \Tilde{u}_i(S,a_i,a'_i)$, and that the strategic move is decided by a weighted average of $\Tilde{u}_i$ values.
    
    We distinguish six possible cases:
    \begin{itemize}
    \item Case A: $\exists s\in S, u_i(a_i’|a_i,s)=1 \text{ and } \forall s\in S, u_i(a_i’|a_i,s)\geq 0$
    \item Case B: $\exists s\in S, u_i(a_i’|a_i,s)=-1 \text{ and } \forall s\in S, u_i(a_i’|a_i,s)\leqslant 0$
    \item Case C: $\exists s\in S, u_i(a_i’|a_i,s)=-1 \text{ and } \exists s\in S, u_i(a_i’|a_i,s)=1 $
    \item Case D: $\forall s\in S, u_i(a_i’|a_i,s)=0 $
    \item Case A.1: $\forall s\in S, u_i(a_i’|a_i,s)=1$
    \item Case B.1: $\forall s\in S, u_i(a_i’|a_i,s)=-1 $
    \end{itemize}
    However, cases A.1 and B.1 can never happen for a non-singleton $S$, for the same reasons as the ones advocated in the proof of Theorem~\ref{cvgone}. We then get
    
    \[\Tilde{u_i}(S,a_i,a'_i) =
    \begin{cases}
        1-\alpha & \text{Case A}, \\
        -\alpha & \text{Case B}, \\
        1-2 \cdot \alpha & \text{Case C}, \\
        0 & \text{Case D}
    \end{cases}\]
    It is clear that $\alpha > \frac{1}{2}$ implies that $\Tilde{u}_i(S,a_i,a'_i) \geq 0$ in cases A only.
    
    Therefore, if \[\alpha \cdot \underline{\mathbb{E}}_{\mathcal{M}}(u_i(a_i’|a_i,s))+ (1-\alpha) \cdot  \overline{\mathbb{E}}_{\mathcal{M}}(u_i(a_i’|a_i,s)) \geq 0\] then
    \[\exists r_i'\leqslant r_i \text{ such that } \exists s\in S_{r_i'}, u_i(a_i’|a_i,s)=1 \text{ and } \forall s\in S_{r_i'}, u_i(a_i’|a_i,s)\geq 0 \]
    
    Which is Case A. The last equivalence comes from the fact that only case A can lead to a positive $\Tilde{u}_i(S,a_i,a'_i)$, and that the criterion is an average of such $\Tilde{u}_i$ values. Therefore, there must exist $r'_i\leqslant r_i$ such that $\Tilde{u}_i(S_{r'_i},a_i,a'_i)=1-\alpha$. If $R^+_i=\{j: |\Tilde{u}_i(S_{j},a_i,a'_i)=1-\alpha\}$ is the set of neighborhood indices in case A, we need for $\sum_{j \in R^+_i}\beta_j^i$ to be large enough for the decision criterion to be positive. In other words, we can accept a move with some case C only if there exists a local dominance move in a smaller neighborhood that receives enough evidence.
    Using the result from section VIII in Meir~\cite{meir2015plurality} which tells us that the convergence holds for any level of uncertainty $r_i$ and any starting point (even non-truthful states) for local dominance strategic behaviors, then we can do a bijection of our strategic moves and get the convergence also.
    \end{proof}
    
    \section{A word on complexity}


Let us briefly situate our model with respect to the computational complexity of computing a manipulation. Prior work has already investigated related questions, and it is useful to distinguish two broad cases. For plurality, when one must evaluate each scenario induced by uncertainty, for instance, under incomplete preferences, this essentially amounts to reasoning over all completions. This is closely related to the possible winner task and is \#P-hard~\cite{conitzer2011dominating}. In particular, under criterion that aggregate across completions (e.g., computing an expectation under a pignistic rule), one cannot avoid considering the full set of completions. By contrast, when the decision criterion is simpler, such as a pessimistic (worst-case) criterion, plurality requires only limited information; it suffices to identify the winner in the relevant  worst-case scenario. In that case, the problem can be rewritten as a flow problem as in Conitzer et al.~\cite{conitzer2011dominating} which can be solved in polynomial time. We may also mention other decision criterion that remain simple enough to be computed in polynomial-time~\cite{dey2018complexity}, such as the optimistic criterion (i.e., the preferred candidate wins in at least one completion) and opportunistic manipulation (i.e., the preferred candidate wins in every completion in which some manipulation is feasible). 

The same dichotomy can be applied to our approaches: if a decision rule requires enumerating all possible completions of a vote, we will end up with \#P-hard problems. This is the case for decision rules involving the pignistic probability, for instance. On the other hand, and provided the belief function mass $\mathcal{M}$ is positive only on a polynomial number of sets, rules that only require identifying one completion remains polynomial. 

In our general model, including these cases as particular cases, we then know that under any criteria that needs to compute every possible scenario then the decision will be also \#P-hard to decide, or worse if we consider quite generic/complex models. However, the case of a pessimistic decision criterion remains solvable in polynomial time as long as the number of sets receiving positive mass (in the case of belief functions) remains polynomial, essentially because it does not require computing probabilities or weights. Thanks to this observation, a standard maximum-flow approach~\cite{conitzer2011dominating} can be applied to obtain a polynomial-time algorithm for manipulation with the pessimistic criterion.

This also opens up an interesting avenue of research: given that in the probabilistic case all decision rules amount to classical expected value, could we weaken the initial probabilistic information, e.g., by considering an inner measure on a coarse partition of linear orders, and take a tractable decision rule (pessimistic and optimistic) in order to approximate intractable probabilistic manipulation rules? To be more precise, if computing an expectation $\mathbb{E}(\cdot)$ for a probability is intractable, can we find a set $\mathcal{P}$ containing this probability such that computing $\underline{\mathbb{E}}(\cdot),\overline{\mathbb{E}}(\cdot)$ is tractable, and where $\underline{\mathbb{E}}(\cdot)$, $\overline{\mathbb{E}}(\cdot)$ are reasonably close to each other? This would let us reach a decision in some cases, and in others, quantify the remaining uncertainty. Note that such a question is made possible only by the fact that our models unify set-based with probabilistic-based decision (respectively known as decision under uncertainty, and decision under risk~\cite{grabisch2016decision}) in a single framework, thereby allowing one to easily transfer notions used in one setting to the other setting. 


\end{document}